\documentclass[letterpaper, 10 pt, conference]{ieeeconf}  % Comment this line out if you need a4paper

\IEEEoverridecommandlockouts                              % This command is only needed if 
\usepackage{graphicx}
\usepackage{textcomp}
\usepackage{amsmath, amssymb}
\usepackage{nicefrac}       % compact symbols for 1/2, etc.
\usepackage{comment}
\usepackage{siunitx}
\usepackage{aligned-overset}
\usepackage{algorithm}
\usepackage[noend]{algpseudocode}
\usepackage{xcolor}
\usepackage{tikz}
\usepackage{cite}
\usepackage{pgfplots}
\usetikzlibrary{positioning}
\let\labelindent\relax
\usepackage[shortlabels]{enumitem}
\usepackage{pgfplots}
\usepackage{wrapfig}
\usetikzlibrary{arrows.meta}
\usetikzlibrary{external}
\usetikzlibrary{graphs}
\usepgfplotslibrary{groupplots}
\usepackage{float}
\usepackage{comment}
\usepackage{subcaption}
\usepackage{todonotes}

\newcommand{\safeopt}{\textsc{SafeOpt}}

\newcommand{\ie}{i\/.\/e\/.,\/~}
\newcommand{\eg}{e\/.\/g\/.,\/~}

\newtheorem{theorem}{\bf Theorem}

\newtheorem{assumption}{\bf Assumption}

\newtheorem{remark}{\bf Remark}
\newtheorem{example}{\bf Example}
\newcommand{\domain}{\mathcal A}

\renewcommand{\P}{\mathbb{P}}

\newcommand{\I}{\mathcal{I}}
\newcommand{\Ig}{\mathcal{I}_{\mathrm{g}}}

\newcommand{\fakepar}[1]{\vspace{1mm}\noindent\textbf{#1.}}
\DeclareMathOperator*{\argmax}{arg\,max\,}
\DeclareMathOperator*{\argmin}{arg\,min\,}

\definecolor{aaltoRed}{RGB}{239,51,64}%
\definecolor{aaltoBlue}{RGB}{0,94,184}%
\definecolor{aaltoGray}{RGB}{140,133,123}%
\definecolor{aaltoOrange}{RGB}{255, 141, 79}%

\DeclareMathOperator*{\R}{\mathbb{R}}

\newcommand{\jj}{{(j)}}
\newcommand{\te}{{t_\mathrm{e}}}
\newcommand{\Te}{{T_\mathrm{e}}}

\newcounter{tool}

\definecolor{magenta}{RGB}{255,0,255}
\definecolor{lightgreen}{RGB}{76, 200, 76}
\definecolor{deeporange}{RGB}{204, 112, 0}
\newcommand{\lline}{$\mathscr{L}.$}
\usepackage[most]{tcolorbox}
\usepackage{xcolor}
\usepackage{changepage}
\usepackage{mathrsfs}
\usepackage{lipsum} % for dummy text
\usepackage[framemethod=tikz]{mdframed}
\usepackage[colorlinks=true, linkcolor=blue, citecolor=blue, urlcolor=blue]{hyperref}

\title{\LARGE \bf
Safe learning-based control \\ via function-based uncertainty quantification}

\author{Abdullah Tokmak$^{1}$, Toni Karvonen$^{2}$, Thomas B.\ Schön$^{3}$, and Dominik Baumann$^{1}$% <-this % stops a space
\thanks{*This research was partially supported by a Tandem Industry Academia Seed funding from the Finnish Research Impact Foundation, the Finnish Automation Society, the Research Council of Finland projects 359183  and 368086, the ELLIS Institute Finland, and the Kjell och Märta Beijer Foundation. The authors acknowledge the use of MIDAS infrastructure of Aalto 
School of Electrical Engineering.
}% <-this % stops a space
\thanks{$^{1}$
Cyber-physical Systems Group, Aalto University, Espoo, Finland {\tt\small firstname.lastname@aalto.fi}}%
\thanks{$^{2}$
 School of Engineering Sciences,
        Lappeenranta--Lahti University of Technology LUT, Lappeenranta, Finland
        {\tt\small toni.karvonen@lut.fi}}%
\thanks{$^{3}$
Department of Information Technology,
        Uppsala University, Uppsala, Sweden
        {\tt\small thomas.schon@uu.se}}%
}

\usepackage{fancyhdr}
\fancypagestyle{firstpage}{%
  \fancyhf{} % clear header/footer
  \fancyfoot[L]{\normalfont \sffamily \scriptsize \lowernote}
}
\fancypagestyle{plain}{%
  \fancyhf{}
  \fancyfoot[L]{\normalfont \sffamily \scriptsize \lowernote}
}
\newcommand{\lowernote}{
\textbf{Camera-ready version version. \textcopyright\; 2026
IEEE.}
To be presented at the Conference on Decision and Control (CDC), 2026.
Personal use of this material is permitted. Permission from IEEE must be obtained for all other uses, in any current or future media,
including reprinting/republishing this material for advertising or promotional purposes, creating new collective works, for resale or redistribution to servers
or lists, or reuse of any copyrighted component of this work in other works.
}
\begin{document}
\newmdenv[innerlinewidth=0.5pt, roundcorner=4pt,linecolor=aaltoBlue,innerleftmargin=6pt,
backgroundcolor=aaltoBlue!10,
innerrightmargin=6pt,innertopmargin=6pt,innerbottommargin=6pt]{mybox}

\maketitle
\thispagestyle{firstpage}
%\thispagestyle{empty}
%\pagestyle{empty}

%%%%%%%%%%%%%%%%%%%%%%%%%%%%%%%%%%%%%%%%%%%%%%%%%%%%%%%%%%%%%%%%%%%%%%%%%%%%%%%%
\begin{abstract}
Uncertainty quantification is essential when deploying learning-based control methods in safety-critical systems.
This is commonly realized by constructing uncertainty tubes that enclose the unknown function of interest---\eg the reward and constraint functions or the underlying dynamics model---with high probability. 
However, existing approaches for uncertainty quantification typically rely on restrictive assumptions that encode smoothness properties of the unknown function, 
such as a known norm in a function space.
Moreover, these methods usually struggle with discontinuities.
In this paper, we model the unknown function as a random function from which independent and identically distributed realizations can be generated.
We then construct uncertainty tubes via the scenario approach that hold with high probability.
Our uncertainty tubes rely solely on sampled realizations and can therefore accommodate discontinuities represented by the sampling model.
We integrate these uncertainty tubes into a safe Bayesian optimization algorithm with which we safely tune control parameters on a real Furuta pendulum.
\end{abstract}
% \begin{IEEEkeywords}
% Bayesian optimization; Gaussian processes; Distributed optimization; Networked multi-agent systems
% \end{IEEEkeywords}

% \input{Figures/structure.tex}

\section{Introduction}\label{sec:introduction}

% When mobile robots enter the real world, they operate in uncertain and highly complex environments.
% Moreover, the dynamics of modern robotic systems, such as soft robots~\cite{yasa2023overview}, are difficult to model accurately from first principles.
% Consequently, machine learning methods are of increasing interest to the control community to improve the performance of such systems, giving rise to the thriving field of learning-based control~\cite{brunke2022safe}.

Learning-based control~\cite{brunke2022safe} can achieve high-performing control policies for complex dynamical systems from experimental data, with reinforcement learning (RL)~\cite{sutton1998reinforcement} as a prominent example. % that aims to obtain high-performing control policies for complex systems from experimental data.
However, to tune control parameters of safety-critical real-world systems, we must~\emph{(i)} ensure sample efficiency and~\emph{(ii)} guarantee safe operation, \ie avoid damage to hardware or harm to human life. Vanilla RL struggles with both requirements. 
Safe Bayesian optimization (BO)~\cite{sui2015safe} offers a suitable framework for sample-efficient control parameter tuning with safety guarantees and has already been successfully applied to complex systems like robots~\cite{berkenkamp2023bayesian}. 
The probabilistic safety guarantees of safe BO are obtained through explicit uncertainty quantification (UQ), typically by constructing high-probability uncertainty tubes around the unknown function of interest; \eg the reward and constraint functions associated with a given control parameter.

\begin{figure}[h]
    \begin{subfigure}{2.75cm}
        \centering
        \includegraphics[width=2.75cm]{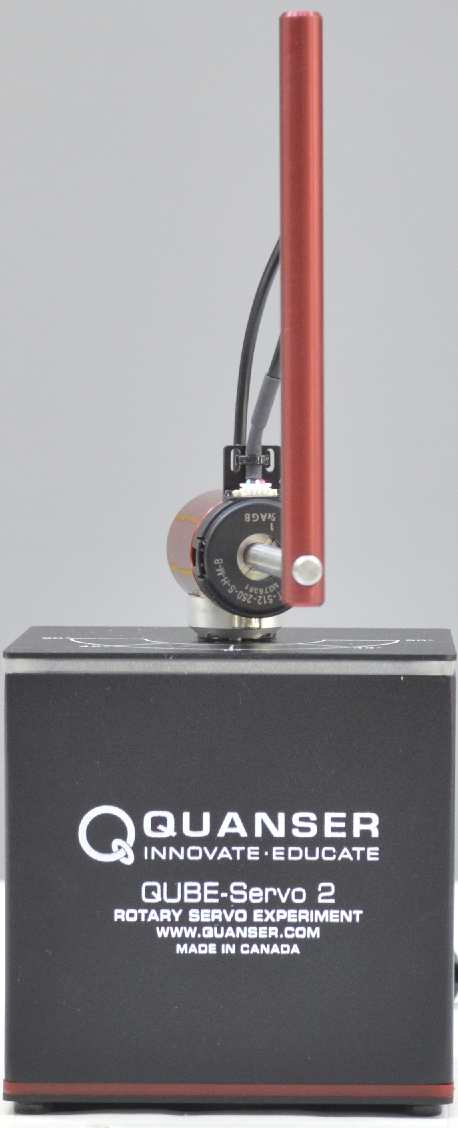}
        \caption{}
        \label{fig:1a}
    \end{subfigure}
\begin{subfigure}{6cm}
    \centering
    {
    \input{Figures/final/exploration_beginning}
    % \vspace*{-3mm}
    \input{Figures/final/exploration_end}}
    \caption{}
    \label{fig:1b}
\end{subfigure}
    \caption{\emph{Control parameter tuning using safe BO on a real Furuta pendulum.}
    Figure~\ref{fig:1a} shows the experimental setup of the Furuta pendulum~\cite{furuta1992swing} and Figure~\ref{fig:1b} illustrates the 
    exploration of the parameter space.
    We first (upper sub-figure) conduct an experiment with the initial safe parameter (magenta diamond).
    Based on the observed reward and constraints, we construct uncertainty tubes, where red and blue denote high and low reward estimates, respectively.
    Using these estimates, we compute a safe set (green hull) that contains only parameters we believe to be safe with high probability.
    We sequentially evaluate new parameters using our safe BO algorithm. 
    After 20 iterations (lower sub-figure), we have explored parts of the domain, expanded the safe set, and refined the uncertainty tubes.
    The cyan square marks the best parameter identified by the safe BO algorithm, which achieves significantly better control performance compared to the initial parameter, as detailed in Section~\ref{sec:furuta}.
    }
    \label{fig:Furuta_intro}
\end{figure}

In this paper, we model the unknown function 
as a random function from which independent and identically distributed (i.i.d.) function realizations can be generated; 
we refer to this technique as \emph{function-based UQ}.
Not only does function-based UQ provide a general modeling framework that captures a wide range of formulations, but it also allows for the use of the so-called scenario approach~\cite{campi2018introduction, campi2018wait, garatti2022complexity}.
The scenario approach is a well-known data-driven decision-making tool in the control community that we use to compute sampling-based uncertainty tubes, which we then integrate into a safe BO algorithm.
Unlike many existing uncertainty tubes in safe BO, our approach requires no smoothness information beyond that encoded by the i.i.d.\ function realizations, and discontinuities can be incorporated straightforwardly through the sampling model.
Although demonstrated here within safe BO for control parameter tuning (Figure~\ref{fig:Furuta_intro}), the proposed uncertainty tubes are applicable to other settings, including model predictive control~\cite{hewing2020learning} and trajectory planning~\cite{mirzaeedodangeh2025safe}.

\fakepar{Contributions}
We make the following contributions.
\begin{adjustwidth}{0.5cm}{}
\begin{enumerate}[label=(C\arabic*)]
    \item We introduce function-based UQ, a flexible and interpretable probabilistic modeling framework for unknown functions.
    \label{co:asm}
    \item Under~\ref{co:asm}, we derive probabilistic uncertainty tubes by using the scenario approach~\cite{campi2018introduction}  and present potentially tighter uncertainty tubes based on the wait-and-judge scenario framework~\cite{campi2018wait}. \label{co:bounds}
    \item We integrate~\ref{co:bounds} into a safe BO algorithm,\footnote{%
    Implementation available at: \url{http://github.com/tokmaka1/CDC-2026}
    } with which we tune the control parameters of a Furuta pendulum without ever incurring a safety violation. \label{co:exp}
\end{enumerate} 
\end{adjustwidth}

\fakepar{Related work}
\safeopt~\cite{berkenkamp2023bayesian, sui2015safe} is arguably the most popular safe BO algorithm.
Its uncertainty tubes are of frequentist, \ie worst-case nature, which is desirable when safety is paramount~\cite{chowdhury2017kernelized, fiedler2021practical, lahr2025optimal, fiedler2024safety, prajapat2025finite}.
However, these frequentist uncertainty tubes depend on a functional norm, specifically, on a tight upper bound on the norm of the unknown function in its corresponding reproducing kernel Hilbert space (RKHS).
Obtaining such a tight upper bound on the RKHS norm of the unknown function is notoriously difficult in practice, which limits the practical applicability of these guarantees~\cite{fiedler2024safety, tokmak2024pacsbo, tokmak2025safe}.
In response to this shortcoming,~\cite{fiedler2024safety} proposes a safe BO algorithm that replaces the RKHS norm with a Lipschitz constant.
Although more interpretable, the Lipschitz constant is likewise typically unknown.
Approaches exist that upper-bound the RKHS norm~\cite{tokmak2024pacsbo, tokmak2025safe, tokmak2025automatic} or the Lipschitz constant~\cite{qin2023conformance, wood1996estimation, lederer2019uniform}. %  with statistical guarantees.
Nonetheless, these existing methods are either heuristic~\cite{tokmak2025automatic, wood1996estimation} or provide statistical guarantees that are \emph{marginal} w.r.t.\ the unknown function, thereby departing from the agnostic frequentist setting~\cite{tokmak2024pacsbo, tokmak2025safe, qin2023conformance, lederer2019uniform}. 

The marginal setting is also prevalent in learning-based control~\cite{srinivas2012information, umlauft2018scenario, lederer2019uniform} and may offer greater modeling flexibility, albeit at the cost of worst-case guarantees.
The seminal paper~\cite[Lemma~5.1]{srinivas2012information} derives uncertainty tubes under the assumption that the unknown function is a sample from a Gaussian process (GP)~\cite{williams2006gaussian} defined on a finite domain.
However, these bounds are highly conservative, as they scale with the cardinality of the finite domain.
Moreover,~\cite[Theorem~3.1]{lederer2019uniform} derives uncertainty tubes for functions that are samples of a GP on compact domains. 
Those uncertainty tubes depend on the Lipschitz constant of the function, which the authors upper-bound probabilistically in~\cite[Theorem~3.2]{lederer2019uniform}.
Reference~\cite{umlauft2018scenario} also works under the GP assumption.
Crucially, the main theoretical result follows from sampling random trajectories that are i.i.d.\ to the unknown trajectory.
Then, the authors use the scenario approach to obtain a performance bound on the unknown function that depends on the performance of the i.i.d.\ sampled functions.

The idea in~\cite{umlauft2018scenario} is closely related to the function-based UQ setting that we formalize, where we assume that i.i.d.\ realizations of the unknown function can be generated.
Nevertheless, we do not restrict ourselves to the Gaussian setting.
Instead, we consider functions that admit a finite linear expansion with fixed, user-chosen basis functions and random coefficients, from which i.i.d.\ realizations can be generated.
This includes flexible choices such as heavy-tailed coefficients or alternative bases like Haar wavelets~\cite{mallat1999wavelet} that introduce discontinuities.
In addition to the induced flexibility, function-based UQ facilitates the use of the classic scenario approach~\cite{campi2018introduction} and the wait-and-judge scenario approach~\cite{campi2018wait} to construct high-probability uncertainty tubes and capture the behavior of the unknown function.

\fakepar{Notation}
The cardinality of a set~$\domain\subseteq\mathbb R^n$ is denoted by~$\lvert\domain\rvert.$
For scalars~$\xi_{i,k}\in\R$, with~$i\in I, k\in K$ and finite index sets such that~$\lvert I \rvert = N_I$ and $\lvert K \rvert = N_K$ we write~$[\xi]_{i\in I, k\in K}\coloneqq [\xi_{i_1, k_1}, \xi_{i_2, k_{1}}, \ldots \xi_{i_{N_I}, k_1}, \ldots, \xi_{i_{N_I}, k_{N_K}}]^\top\in\mathbb R^{N_I\cdot N_K}$.
Correspondingly, for any function~$h_i{:}\; \domain\to\R$, with~$\lvert\domain\rvert<\infty$ and~$i\in I, \lvert I\rvert =N_I$, we define~$[h_i(a)]_{i\in I, a\in\domain}\coloneqq[h_{1}(a_1),\ldots, h_{N_I}(a_1), \ldots, h_1(a_{\lvert \domain\rvert}),\ldots, h_{ N_I}(a_{\lvert \domain\rvert})]^\top \in \mathbb R^{\lvert\domain\rvert\cdot N_I}$.
For any $m\in\mathbb N$, we define $[m]\coloneqq \{1,\ldots,m\}$ and denote by~$\mathbf 1_m=[1,\ldots,1]\in\mathbb R^m$.

\section{Problem formulation}\label{sec:problem_setting}
Let us now make our problem formulation precise.
We assume that a parametrized control policy exists such that,
% We assume that the control policy is parameterized such that, 
at any iteration~$t\geq 1$, we can conduct an experiment following the policy parameter~$a_t\in\domain\subseteq \mathbb R^n$, where~$\domain$ is a finite domain, potentially obtained by discretization, \ie$\lvert\domain\rvert=N<\infty$.
We aim to maximize a reward function $h_0\colon\domain\to\R$ while guaranteeing safety, \ie satisfying constraints for every conducted experiment.
This can be written as the constrained optimization problem
\begin{align}\label{eq:opt}
\max_{a\in\domain} h_0(a) \quad \text{subject to} \quad h_i(a_t)\geq \underline h_i, \forall t\geq 1, \forall i\in\Ig,
\end{align}
where $h_i\colon \domain\to \mathbb R$ for~$i\in\Ig\subseteq\mathbb N$ are the constraints and~$\underline h_i$ are the safety thresholds.
We also define~$\I\coloneqq \{0\}\cup\Ig$.
The reward function quantifies a policy parameter's control performance, while the constraint functions may map a policy parameter to the minimum distance to an obstacle over a trajectory.
At any iteration~$t\geq 1$, an experiment with policy parameter~$a_t\in\domain$ yields noisy observations 
$y_{i,t}\coloneqq h_i(a_t)+\epsilon_{i,t}$ for each $i \in \I$. %;
We collect the chosen parameters and the associated observations in the vectors $a_{1:t}\coloneqq [a_{t^\prime}]_{t^\prime\in[t]}$ and $y_{1:t}\coloneqq [y_{i,t^\prime}]_{i\in\I,t^\prime\in [t]}$, respectively.
The data set is denoted by~$\mathcal D_t\coloneqq (a_{1:t}, y_{1:t})$.
%we collect the dataset of evaluated policy parameters and corresponding noisy observations as~$D_t\coloneqq \{(a_{t^\prime}, y_{i,t^\prime})\}_{t^\prime=1.\ldots,t; i\in\I}$.
The stochasticity of the observation noise~$\epsilon_{i,t}$ is as follows.
\vspace{0.1cm}
\begin{assumption}[Observation noise]\label{asm:noise}
   For each~$i\in\I$ and for any~$t\geq 1$, the observation noise~$\epsilon_{i,t}$ is distributed according to a probability distribution~$\P_\epsilon$, from which we can generate i.i.d.\ samples. 
    \end{assumption}
    \vspace{0.1cm}
Assumption~\ref{asm:noise}  accommodates general noise models including heteroscedastic and heavy-tailed distributions with possibly unbounded support, see~\cite{tokmak2026safe}.

To design an algorithm that solves~\eqref{eq:opt}, it is essential to characterize the behavior of the unknown functions~$h_i$ such that control parameters~$a_t$ can be selected that satisfy the prescribed safety constraints with high probability.
This motivates the construction of uncertainty tubes of the form %\footnote{The probability measure~$\P$ in~\eqref{eq:tube} contains randomness of the noise~$\epsilon_{i,t}$ and the randomness of the unknown function~$h_i$ conditioned on the observed data; we make~$\P$ precise in Section~\ref{sec:posterior}. }
\begin{align}\label{eq:tube}
\mathbb P
\left[\exists i\in \I, a\in\domain\,{:}\, h_i(a) \not\in [\ell_{i,t}(a),u_{i,t}(a)]
\right] \leq \nu,
\end{align}
\ie we require that the unknown functions~$h_i$ are contained within the lower bounds~$\ell_{i,t}$ and upper bounds~$u_{i,t}$ with a user-chosen probability level~$1-\nu$, for all parameters in the domain.
Establishing bounds as in~\eqref{eq:tube} requires a probabilistic modeling framework and standing assumptions, which we formalize in Section~\ref{sec:function-based-UQ}.
In Section~\ref{sec:tube}, we solely focus on satisfying the uncertainty tube condition~\eqref{eq:tube}, before returning to Problem~\eqref{eq:opt} and tackling it via a safe BO algorithm in Section~\ref{sec:BO}.
We present experiments in Section~\ref{sec:exp} and conclude the paper in Section~\ref{sec:conclusion}.

\section{Function-based uncertainty quantification}\label{sec:function-based-UQ}

We first introduce a probabilistic model for the unknown function~$h_i$ (Section~\ref{sec:modeling}).
Then, we demonstrate how to generate function scenarios that capture the variability of the unknown function~$h_i$ given the observed data (Section~\ref{sec:posterior}).

\subsection{Probabilistic function modeling}\label{sec:modeling}
To construct uncertainty tubes of the form~\eqref{eq:tube}, it is necessary to impose assumptions on the unknown function~$h\coloneqq [h_i(a)]_{i\in\I,a\in\domain}$.
As mentioned in Section~\ref{sec:introduction}, a popular assumption of the
learning-based control community~\cite{umlauft2018scenario, lederer2019uniform, srinivas2012information} is to assume that~$h$ is a sample from a zero-mean GP with a 
kernel~$k\colon \domain\times\domain\to\mathbb R$.
In the finite domain setting with~$\domain = \{a_1, \ldots, a_N\}$, a sample~$h_i$ of a zero-mean GP can be written as $h_i(a)=\sum_{r=1}^N c_{i,r} \varphi_{i,r}(a)$ with basis functions $\varphi_{i,r}(a)=k(a_r, a)$ and jointly Gaussian $c\coloneqq [c_{i,r}]_{i\in\I, r\in [N]}$~\cite{williams2006gaussian}.
Motivated by the Gaussian representation, we generalize this idea by allowing more general coefficient distributions and arbitrary fixed basis functions.
Thus, we can incorporate prior knowledge about the underlying function through the choice of coefficient distributions and basis functions in a flexible, structured, and interpretable manner.
In this way, we can straightforwardly approximate diverse function classes, including structured discontinuities such as Haar wavelets (Section~\ref{sec:admissible}).
Representing such functions is non-trivial with GPs, even when the underlying basis is known, since the corresponding kernels may be difficult to express in closed-form.

\vspace{0.1cm}
\begin{assumption}[Unknown function]\label{asm:prob}
    For each $i \in \I$, the unknown function $h_i$ is a random function of the form 
        $h_i(a) = \sum_{r=1}^N c_{i,r} \varphi_{i,r}(a)$,
where $\varphi_{i,r}$ are fixed basis functions and the random coefficient vector~$c$ has probability distribution~$\P_\mathrm{P}$.
At each iteration~$t\geq 1$, after observing~$\mathcal D_t$, $h$ has distribution~$\P$.%\coloneqq [\alpha_{i,r}]_{i\in\I,r\in[N]}$ 
%has probability distribution~$\mathbb P_\mathrm{P}$, from which we can generate i.i.d.\ draws.
\end{assumption}
\vspace{0.1cm}
Assumption~\ref{asm:prob} allows us to recover the finite-domain GP representation by sampling jointly Gaussian coefficients and setting the basis functions to the kernel functions.
It also enables representing other stochastic processes through suitable choices of coefficient distributions and basis functions.

\subsection{Sampling data-consistent function scenarios}\label{sec:posterior}

To capture the behavior of the unknown function~$h$ after observing the data~$\mathcal D_t$, we generate i.i.d.\ random functions~$\rho_t^\jj\coloneqq [\rho_{i,t}^\jj(a)]_{i\in\I, a\in\domain}$ with
\begin{align}\label{eq:function_scenarios}
    \rho_{i,t}^\jj(a) \coloneqq \sum_{r=1}^N \widetilde c_{i,r}^\jj \varphi_{i,r}(a),
\end{align}
which we refer to as \emph{scenarios}.
We obtain the coefficients in~\eqref{eq:function_scenarios} by solving 
\begin{align}\label{eq:interpolation}
\widetilde c^\jj = \argmin_{c} \|c-\widetilde c_{\mathrm P}^\jj\|_2^2 \quad \text{s.t.} \quad \Phi c = y_{1:t}-\widetilde\epsilon_{1:t}^\jj.
\end{align}
In~\eqref{eq:interpolation},~$\Phi\in\R^{t\lvert\I\rvert \times N\lvert\I\rvert}$  is a generalized Vandermonde matrix induced by the basis functions $\varphi_{i,r}(a_{t^\prime})$, $t^\prime\in[t], i\in \I, r\in [N]$. 
Further, we define~$\widetilde\epsilon_{1:t}^\jj\coloneqq [\widetilde\epsilon_{i,t^\prime}^\jj]_{i\in\I,t^\prime\in[t]}$.
We sample the noise realization~$\widetilde\epsilon_{i,t}^\jj$ according to~$\mathbb P_\epsilon$ (Assumption~\ref{asm:noise}) and the coefficients~$\widetilde c_\mathrm{P}^\jj$ according to~$\P_\mathrm{P}$ (Assumption~\ref{asm:prob}).
The minimum-norm update in~\eqref{eq:interpolation} ensures that the scenarios in~\eqref{eq:function_scenarios} are consistent with the observed data~$\mathcal D_t$, \ie $\rho_{i,t}^\jj(a_{t^\prime})=y_{i,t^\prime}-\widetilde\epsilon_{i,t^\prime}^\jj$.
The solution of~\eqref{eq:interpolation} and hence the coefficients~$\widetilde c^\jj$ in~\eqref{eq:function_scenarios} are given by
\begin{align}\label{eq:coefficients}
     \widetilde c^\jj = \widetilde c_\mathrm{P}^\jj + \Phi^\top (\Phi\Phi^\top)^{-1}(y_{1:t}-\widetilde\epsilon_{1:t}^\jj- \Phi\widetilde c_\mathrm{P}^\jj).
\end{align}
Throughout this article, we assume that~$\Phi\Phi^\top$ is invertible.
In practice, to ensure numerical stability, it may be necessary to add a regularization constant to the matrix inversion in~\eqref{eq:coefficients}, resulting in the Tikhonov-regularized version of~\eqref{eq:interpolation}. %; in place of the inversion, one may also use a Cholesky decomposition.

% The sampling procedure above induces a sampling law for the scenarios, which we now define.
The construction above induces a sampling law for the scenarios. We next formalize this law and establish that, under Assumption~\ref{asm:prob}, the resulting scenarios enable the construction of uncertainty tubes of the form~\eqref{eq:tube}.
Let~$\P_0$ be a joint law of $(\widetilde c_\mathrm{P}, \widetilde \epsilon_{1:t})$ with marginals~$\P_\mathrm{P}$ and~$\P_\epsilon$, respectively.
Consider the composition~$T_{\mathcal D_t}=T_{2} \circ T_{1,\mathcal D_t}$.
The transformation~$T_{1,\mathcal D_t}$ is given by~\eqref{eq:coefficients} and maps the realized~$\widetilde \epsilon_{1:t}^\jj$ and~$\widetilde c_\mathrm{P}^\jj$ from~$\P_0$ to data-consistent coefficients; after observing~$\mathcal D_t$,~$T_{1,\mathcal D_t}$ is affine and hence continuous, which implies that it is measurable~\cite[Exercise~1.3.4~\emph{(i)}]{durrett2019probability}.
The transformation~$T_2$ is given by~\eqref{eq:function_scenarios} and maps the data-consistent coefficients~$\widetilde c^\jj$ to the scenarios~$\rho_{i,t}$; since~$T_2$ is linear, it is also measurable.
Therefore, the composition~$T_{\mathcal D_t}$ is measurable~\cite[Theorem~1.3.4]{durrett2019probability} and the pushforward measure~$\P\coloneqq \P_0\circ T_{\mathcal D_t}^{-1}$---according to which the scenarios~$\rho_t^\jj$ are sampled---exists~\cite[Section~1.2]{durrett2019probability}.
Moreover, independent samples~$(\widetilde c_\mathrm{P}, \widetilde \epsilon_{1:t})$ drawn according to~$\P_0$ remain independent after the same measurable transformation~$T_{\mathcal D_t}$ is applied~\cite[Theorem~2.1.10]{durrett2019probability}.
Hence, the scenarios~$\rho_t^\jj$ are i.i.d.\ according to~$\P$.
By Assumption~\ref{asm:prob}, given~$\mathcal D_t$, $h$ is distributed according to~$\P$ and the construction of the scenarios~$\rho_t^\jj$ is independent of~$h$.
Altogether, after observing~$\mathcal D_t$, the scenarios~$\rho_{t}^\jj$ and~$h$ are i.i.d.\ according to~$\P$.

\begin{remark}[Data-consistency vs.\ Bayesian posterior]\label{re:Bayesian}
    In this work, we adopt~$\P$ as outlined in this section, which yields a principled, data-consistent, and computationally efficient generation of scenarios.
    The scenarios are sufficiently variable to yield meaningful and reliable uncertainty tubes, as we will demonstrate both on synthetic examples (Sec.~\ref{sec:tube}) and real-world hardware (Sec.~\ref{sec:furuta}).
Crucially,~$\P$ encodes a \emph{modeling assumption} on~$h$ after observing~$\mathcal D_t$ under Assumption~\ref{asm:prob}, and the resulting uncertainty tubes~\eqref{eq:tube} and their associated guarantees depend on this modeling choice of~$\P$.
In general,~$\P$ as constructed in this section differs from the Bayesian posterior of~$h$ after observing~$\mathcal D_t$. %; the Bayesian posterior may be used instead whenever i.i.d.\ realizations can be generated from it.
Our forthcoming uncertainty tubes require only that the unknown function and the scenarios are i.i.d.; any law from which we can sample is admissible, and the argument is not tied to a linear expansion~\eqref{eq:function_scenarios} or to the minimum-norm map~\eqref{eq:interpolation}.
\end{remark}

\begin{remark}[Comparison to linear regression]
The finite expansion~\eqref{eq:function_scenarios} is the model of classical linear regression.
A least-squares fit of that model yields a point estimate without uncertainty tubes.
As already noted in Remark~\ref{re:Bayesian}, our uncertainty tubes require only that the unknown function and the scenarios are i.i.d., and are not tied to the specific implementation in this section.
\end{remark}

\section{Derivation of uncertainty tubes}\label{sec:tube}
In Section~\ref{sec:scenario_tube}, we establish the containment of~$h$ within uncertainty tubes, \ie we satisfy~\eqref{eq:tube} with high probability using the classical scenario approach~\cite{campi2018introduction}.
Section~\ref{sec:wait} derives potentially tighter bounds via the wait-and-judge framework~\cite{campi2018wait}.

\subsection{Uncertainty tubes via scenario programming}\label{sec:scenario_tube}
The availability of i.i.d.\ scenarios advocates the use of the scenario approach~\cite{campi2018introduction}.
In the spirit of the scenario approach, we formulate the computation of the tubes~$\ell_{t}\coloneqq[\ell_{i,t}(a)]_{i\in\I,a\in\domain}$ and~$u_{t}\coloneqq[u_{i,t}(a)]_{i\in\I,a\in\domain}$ to satisfy~\eqref{eq:tube} as a convex optimization problem, also referred to as the \emph{scenario program}.
% Before stating the scenario program, let us first denote by~$\ell_t, u_t\in\mathbb R^{N\cdot \lvert \I\rvert}$ 
% the vectors collecting the values of~$\ell_{i,t}(a)$ and~$u_{i,t}(a)$ for all~$a\in\domain$ and all~$i\in\I$, \ie
% $
%     \ell_t\coloneqq [\ell_{t,0}(a_1),\ldots,\ell_{t,\lvert \Ig\rvert}(a_N)]^\top,
%  u_t\coloneqq [u_{t,0}(a_1), \ldots,u_{t,\lvert \Ig\rvert }(a_N)]^\top,
% $
% and define a collected scenario as~$\rho_t^\jj\coloneqq [
% \rho_{0,t}^\jj(a_1),\ldots,\rho_{\lvert\Ig\rvert,t}^\jj(a_N)
% ]^\top\in\mathbb R^{N\lvert\I\rvert}$.
At each iteration~$t\geq 1$, we compute~$m_t\in\mathbb N$ scenarios~$\rho_{i,t}^\jj$ as outlined in Section~\ref{sec:posterior} and solve
\begin{align}\label{eq:scenario_opt}
    \min_{[\ell_t^\top, u_t^\top]^\top\in\mathbb R^{2N\lvert\I\rvert}} %&\sum_{a\in\domain} \sum_{i\in\I} u_{i,t}(a) - \ell_{i,t}(a)  
    &\mathbf 1_{N\lvert I\rvert}^\top (u_t-\ell_t)
    \\
    \text{s.t.} \; u_{i,t}(a) &\geq \rho_{i,t}^{(1)}(a), \quad \forall a\in\domain, i\in\I, %j\in[1,m_t] 
    \nonumber \\
    &\vdots \nonumber \\
    \ell_{i,t}(a) &\leq \rho_{i,t}^{(m_t)}(a), \quad \forall a\in\domain, i\in\I, %j\in[1,m_t]. 
    \nonumber
    \end{align}
\ie we compute the tightest possible tubes~$u_t, \ell_t$ that enclose all scenarios uniformly over the domain~$\domain$ and across all functions~$i\in\I$.
The solution of~\eqref{eq:scenario_opt} can be obtained in closed form and corresponds to the point-wise minimum and maximum of the sampled scenarios, \ie %for all~$i\in\I,$
\begin{align}\label{eq:ell_u}
    \ell_{t}(a) &\coloneqq \left[\min\nolimits_{j\in[m_t]}\rho_{i,t}^\jj(a)\right]_{i\in\I} \\
    u_{t}(a) &\coloneqq\left[\max\nolimits_{j\in[m_t]}\rho_{i,t}^\jj(a)\right]_{i\in\I}. \nonumber
\end{align}

The following theorem derives a \emph{sufficient} number of scenarios~$m_t$ such that the bounds in~\eqref{eq:ell_u} satisfy~\eqref{eq:tube} with confidence~$1-\kappa$ at each iteration~$t$  for a user-defined~$\kappa\in(0,1)$.
The two levels of probability arise from the randomness in the sampled scenarios (outer layer) and the generalization of the resulting solution conditioned on the sampled scenarios (inner layer)~\cite[Section~1.2.5]{campi2018introduction}.
The outer layer is taken w.r.t.\ the product measure~$\P^{m_t}\coloneqq \otimes_{q=1}^{m_t}\P$ and the inner layer w.r.t.~$\P$.

\vspace{0.1cm}
\begin{theorem}[Satisfying~\eqref{eq:tube} via classic scenario theory]\label{th:classic_scenario}
    Let Assumptions~\ref{asm:noise} and~\ref{asm:prob} hold, select any  probability level~$\nu\in(0,1)$, any confidence level~$\kappa\in(0,1)$.
    At any iteration~$t\geq 1$, generate~$m_t\in\mathbb N$ i.i.d.\ scenarios~$\rho_t^\jj$ such that~$m_t$ satisfies\footnote{%
    Note that~\eqref{eq:m_t_scaling} corresponds to the cumulative function of a beta distribution.
    Hence, solving for~$m_t$ is computationally inexpensive using standard numerical libraries and bisection.
   Also note that~\cite[Theorem~1.3]{campi2018introduction} presents an explicit upper bound on~$m_t$ that satisfies~\eqref{eq:m_t_scaling}; however, we retain the implicit expression for~$m_t$ in~\eqref{eq:m_t_scaling} for its tightness and interpretability.
    }
    \begin{align}\label{eq:m_t_scaling}
    \sum_{r=0}^{2N\lvert \I \rvert-1}{ m_t \choose r}
    \nu^r (1-\nu)^{ m_t-r}\leq \kappa,
    \end{align}
    and compute the uncertainty tubes as in~\eqref{eq:ell_u}.
    Then,~\eqref{eq:tube} holds at each~$t\geq 1$ with confidence at least~$1-\kappa$, \ie the unknown functions~$h_i$ are uniformly contained in the uncertainty tubes with the prescribed accuracy for all~$i\in\I.$ %, t\geq 1$.
\end{theorem}
\vspace{0.1cm}

\begin{proof}
The proof is based on the generalization theorem of the scenario approach~\cite[Theorem~3.7]{campi2018introduction}.
The constraints of the scenario program~\eqref{eq:scenario_opt} enforce at each iteration~$t\geq 1$ that~$\rho_{i,t}^\jj(a) \in [\ell_{i,t}(a), u_{i,t}(a)]$, uniformly for all~$a\in\domain$, all~$j\in [m_t]$, and all~$i\in\I$.
Since the unknown function and the scenarios are i.i.d.\ under~$\P$ (see Section~\ref{sec:posterior}), the solution of the scenario program~\eqref{eq:ell_u} generalizes to the unseen scenario~$h$. 
Hence, at each iteration~$t\geq 1$, we have that~$\P^{m_t}\big[\P[\exists i\in \I, \exists a\in\domain: h_i(a)\not\in [\ell_{i,t}(a),u_{i,t}(a)]]\leq \nu\big]\geq 1-\kappa$.
\end{proof}

Theorem~\ref{th:classic_scenario} provides a convenient one-shot procedure to determine a sufficient number of scenarios~$m_t$ to satisfy~\eqref{eq:tube} with high probability. 
However, the left-hand side of~\eqref{eq:m_t_scaling} scales with the discretization size~$N$, which may result in a prohibitively large number of scenarios~$m_t$ to satisfy small confidence parameters~$\kappa$.
The following example illustrates this effect and introduces the notion of \emph{support scenarios}, which can be exploited to reduce the required number of scenarios and thereby tighten the uncertainty tubes.

% \begin{figure}
%     \centering
% \input{Figures/final/scenario_apprach_toy_trig}
%     \caption{\emph{Uncertainty tubes established in Theorem~\ref{th:classic_scenario}}.
%     The unknown function~$h_0$ (blue) is contained in the uncertainty tubes (orange). 
%     The light-orange functions show the~$s_t\leq m_t$ \emph{support scenarios} that contribute to the solution of~\eqref{eq:scenario_opt}.
%     }
%     \label{fig:support}
% \end{figure}

\begin{figure}
        \centering
\input{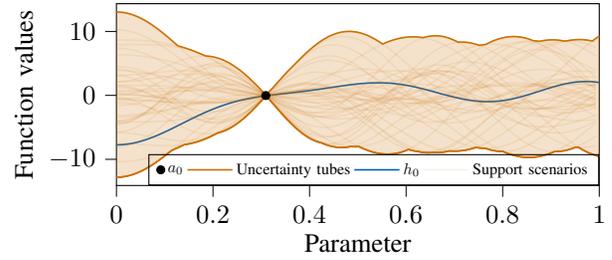}
    \caption{\emph{Uncertainty tubes established in Theorem~\ref{th:classic_scenario}}.
    The unknown function~$h_0$ (blue) is contained in the uncertainty tubes (orange). 
    The light-orange functions show the~$s_t\leq m_t$ \emph{support scenarios} that contribute to the solution of~\eqref{eq:scenario_opt}.
    }
        \label{fig:support}
\end{figure}

\vspace{0.1cm}
\begin{example}\label{ex:1}
Consider a single function (\ie $\Ig=\I=\{0\})$ at iteration~$t=1$, defined on the domain~$\domain=[0,1]$, which is discretized to~$N=1000$ equidistant points.
Further, let the probability level be~$\nu=10^{-1}$ and the confidence level be~$\kappa=10^{-3}$.
We generate scenarios by sampling the coefficient vector from a normal distribution with zero mean and standard deviation $0.1$, see~\eqref{eq:coefficients}.
The basis functions are chosen as trigonometric functions on the discretized domain~$\domain$.
In particular, we choose~$\varphi_{0,r}(a)=\cos\big(
0.005\pi (r-1) a
\big)$, $r\in[N]$.
The observation noise (Assumption~\ref{asm:noise}) is assumed to be from a uniform distribution, \ie $\epsilon_t\sim \mathcal U(-\delta, \delta)$ with~$\delta=0.1$.
By solving~\eqref{eq:m_t_scaling}, we obtain that~$m_t=\num{
21338}$ scenarios are required to compute sufficiently accurate uncertainty tubes.
As such, Theorem~\ref{th:classic_scenario} can be computationally demanding, as it requires generating a large number of function scenarios.
Moreover, a large~$m_t$ can lead to conservative bounds, as illustrated in Figure~\ref{fig:support}.
However, only~$s_t=41$ of the generated scenarios---termed as the \emph{support scenarios}---actually determine the resulting uncertainty tubes.
\end{example}

\subsection{Possibly tighter uncertainty tubes via support scenarios}\label{sec:wait}
Following~\cite[Definition~2]{campi2018wait}, we consider support scenarios as scenarios~$\rho_t^\jj$ whose removal changes the solution of the scenario program~\eqref{eq:scenario_opt}.\footnote{Almost-sure non-degeneracy~\cite[Assumption~2]{campi2018wait} is a standing assumption throughout this paper.}
By Helly's lemma~\cite[Lemma~5.3]{campi2018introduction}, the number of support scenarios~$s_t$ of the convex scenario program~\eqref{eq:scenario_opt} is bounded by the cardinality of the decision variable~$[\ell_t^\top,u_t^\top]^\top\in\mathbb R^{2N\lvert\I\rvert}$, which appears on the left-hand side of~\eqref{eq:m_t_scaling}.
However, most problems adhere to a significantly smaller support size, as demonstrated in Example~\ref{ex:1}.
Although solving the scenario program~\eqref{eq:scenario_opt} and observing~$s_t$ support scenarios is not equivalent to working in the dimension~$s_t$, the gap between these two quantities is quantifiable and has given rise to the wait-and-judge scenario approach~\cite{campi2018wait}.
Following the wait-and-judge framework, we can first wait, \ie solve Problem~\eqref{eq:scenario_opt}, and then a posteriori judge, \ie count the number of support scenarios~$s_t$.

\begin{algorithm}
\begin{algorithmic}[1]
    \Require $\kappa$, $\nu$, $\mathcal D_t$, $\mathbb P_\mathrm{P}, \mathbb P_\epsilon$
\State \textbf{Init:} $s_t \gets 1$, $\bar m_t\gets 0$
    \For{$q=1,2,3,\ldots$}
    \State With~$\tau=1-\nu$, find smallest integer~$m_{t,q}\geq s_t$ s.t.
    \[
     \frac{6\kappa}{(m_{t,q}+1)\pi^2q^2}\sum_{r=s_t}^{m_{t,q}} {r\choose s_t} \tau^{r-s_t} \geq  {m_{t,q}\choose s_t} 
    \tau^{m_{t,q}-s_t}
    \]
    \State Sample~$m_{t,q}$ scenarios~$\rho_t^\jj$, set $\bar m_t=\bar m_t+m_{t,q}$ % \Comment{Sec.~\ref{sec:posterior}}
    \State Compute uncertainty tubes $\ell_t, u_t$ \Comment{\eqref{eq:scenario_opt}, \eqref{eq:ell_u}}
    \State Count number of support scenarios~$s_t$ of~\eqref{eq:ell_u}
   \If{$m_{t,q}=s_t$} continue
    \EndIf
    \State Solve for~$\tau\in(0,1)$ such that
    \[
        \frac{6\kappa}{(m_{t,q}+1)\pi^2q^2}\sum_{r=s_t}^{m_{t,q}} {r\choose s_t} \tau^{r-s_t} -  {m_{t,q}\choose s_t} 
    \tau^{m_{t,q}-s_t} = 0.
    \]
    \If{$\tau\geq 1- \nu$} \Comment{Tubes sufficiently accurate}
    \State \Return $\ell_t, u_t, \bar m_t, m_t=m_{t,q}$
    \EndIf
    \EndFor
\end{algorithmic}
    \caption{Wait-and-judge framework to satisfy~\eqref{eq:tube}}
    \label{alg:wait}
\end{algorithm}

Algorithm~\ref{alg:wait} summarizes the procedure to compute sufficiently accurate uncertainty tubes~$\ell_t$ and~$u_t$ based on the wait-and-judge framework.
In contrast to the classic scenario bounds obtained via Theorem~\ref{th:classic_scenario}, the wait-and-judge approach does not provide a one-shot procedure for computing the required number of scenarios~$m_t$.
Therefore, we first need to initialize the number of support scenarios~$s_t$ (\lline~1). 
Then, we determine the sufficient number of scenarios~$m_{t,q}$ under the condition that~$s_t$ is the support size of the problem.
Then, we sample~$m_{t,q}$ scenarios~$\rho_t^\jj$ as detailed in Section~\ref{sec:posterior} (\lline~4) before formalizing the scenario program~\eqref{eq:scenario_opt}  and obtaining its solution~\eqref{eq:ell_u} (\lline~5).
Next, we count the actual number of support scenarios~$s_t$, \ie the number of scenarios contributing to the solution~\eqref{eq:ell_u} (\lline~6), which can be done efficiently by counting the arguments attaining the extrema in~\eqref{eq:ell_u}. 
If~$s_t=m_{t,q}$, the upcoming root-finding  problem (\lline~8) does not have a solution. In this case, we restart the procedure with the updated~$s_t$ (\lline~3).
Else, we obtain~$\tau\in (0,1)$ by solving a root-finding problem (\lline~8), \eg using bisection~\cite[Appendix~2]{campi2018wait}. %\footnote{%
If~$\tau\geq 1-\nu$ (\lline~9--10), then we have obtained sufficiently accurate uncertainty tubes that satisfy~\eqref{eq:tube}. % with confidence at least~$1-\kappa_t$, where~$\kappa_t\coloneqq\nicefrac{6\kappa}{\pi^2 t^2}$, with user-chosen~$\kappa\in (0,1)$.
Otherwise, we continue the procedure with the updated~$s_t$.
Algorithm~\ref{alg:wait} determines~$m_t$ via wait-and-judge theory with freshly drawn scenarios. 
For a more detailed discussion on sequential scenario optimization, including on how to reuse drawn scenarios, we refer to~\cite{garatti2022complexity}.
% Note that the maximum number of scenarios is bounded by~$\overline m_t$ obtained from Theorem~\ref{th:classic_scenario} (\lline~2), which corresponds to the classic scenario bound using the support dimension derived from Helly’s lemma and thus provides an a priori stopping criterion for the algorithm.

Next, we prove that Algorithm~\ref{alg:wait} indeed yields sufficiently accurate uncertainty tubes.

\vspace{0.1cm}
\begin{theorem}[Satisfying~\eqref{eq:tube} via wait-and-judge]\label{th:uncertainty_tubes}
    Let Assumptions~\ref{asm:noise} and~\ref{asm:prob} hold, select any  probability level~$\nu\in(0,1)$ and any confidence level~$\kappa\in(0,1)$. %, and define~$\kappa_t\coloneqq \nicefrac{6\kappa}{\pi^2t^2}$.
    At any iteration~$t\geq 1$, compute~$\ell_t$ and~$u_t$ according to Algorithm~\ref{alg:wait}.
    Then,~\eqref{eq:tube} holds at each~$t\geq 1$ with confidence at least~$1-\kappa$, \ie the unknown functions~$h_i$ are contained in the uncertainty tubes with the prescribed accuracy for all~$i\in\I$. 
\end{theorem}
\vspace{0.1cm}
\begin{proof}
The proof is based on the generalization theorem of the wait-and-judge scenario approach~\cite[Theorem~4]{campi2018wait},~\cite[Theorem~1]{garatti2022complexity}.
Analogous to the proof of Theorem~\ref{th:classic_scenario}, the solution of the scenario program~\eqref{eq:ell_u} generalizes to~$h$.
At each iteration~$t\geq 1$, we have 
$\bar\P_t\big[\P[\exists i\in \I, \exists a\in\domain\colon h_i(a)\not\in [\ell_{i,t}(a),u_{i,t}(a)]]> 1-\tau]\leq \sum_{q=1}^{\infty}\frac{6\kappa}{q^2\pi^2}=\kappa$.
We use~$\bar\P_t$ to denote the outer probabilities w.r.t.\ all scenario sampling performed in Algorithm~\ref{alg:wait} at iteration~$t$.
Since we exit the algorithm only if~$\tau  \geq 1- \nu$ (\lline~9), \ie $\nu\geq 1-\tau$, we have $\bar\P_t\big[\P[\exists i\in \I, \exists a\in\domain\colon h_i(a)\not\in [\ell_{i,t}(a),u_{i,t}(a)]]\leq \nu]\geq 1-\kappa$, from which the  satisfaction of~\eqref{eq:tube} with confidence at least~$1-\kappa$ directly follows.
\end{proof}

Theorem~\ref{th:uncertainty_tubes} proves that computing uncertainty tubes by following the wait-and-judge framework outlined in Algorithm~\ref{alg:wait} satisfies~\eqref{eq:tube} with user-defined accuracy. We now revisit Example~\ref{ex:1} to illustrate that the wait-and-judge approach can indeed yield tighter bounds in practice.

% \begin{figure}
%     \centering
% \input{Figures/final/bounds_comparison_trig}
%     \caption{\emph{Uncertainty tubes established in Alg.~\ref{alg:wait}/Thm.~\ref{th:uncertainty_tubes}.}
%     The uncertainty tubes following the wait-and-judge framework (gray) are tighter than the classic scenario approach (orange).
%     }
%         \label{fig:example_continued_appendix}
% \end{figure}

\begin{figure}
    \centering
\input{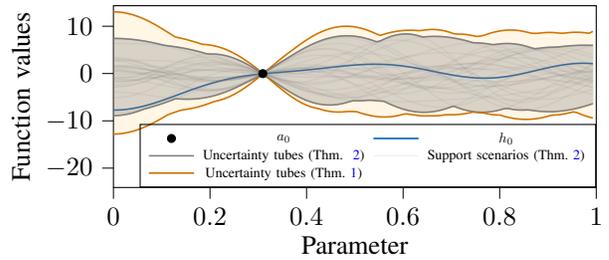}
    \caption{\emph{Uncertainty tubes established in Alg.~\ref{alg:wait}/Thm.~\ref{th:uncertainty_tubes}.}
    The uncertainty tubes following wait-and-judge (gray) are mostly tighter than the classic scenario approach (orange).
    }
        \label{fig:example_continued}
\end{figure}

\vspace*{0.1cm}
\setcounter{example}{0}
\begin{example}[revisited]
Executing Algorithm~\ref{alg:wait} for the setting in Example~\ref{ex:1} yields~$m_t=\num{674}$ scenarios with~$s_t=\num{28}$ support scenarios, resulting in sufficiently accurate uncertainty tubes; see Figure~\ref{fig:example_continued}.
The procedure in Algorithm~\ref{alg:wait} required a total of~$\bar m_t=\num{1877}$.
This contrasts sharply with the $m_t=\num{
21338}$ scenarios required when following the classical scenario approach according to Theorem~\ref{th:classic_scenario}.
The wait-and-judge framework not only decreases the required number of scenarios but it also yields tighter uncertainty bounds over most of the domain for this experiment.
\end{example}

% \section{Warm-up: Scalar uncertainty quantification}\label{sec:WU}

%\input{Sections/scalar}

\section{Applications to safe Bayesian optimization}\label{sec:BO}

After deriving the uncertainty tubes~$\ell_t, u_t$, we finally integrate them into a safe BO algorithm to solve Problem~\eqref{eq:opt} efficiently and effectively.
Our safe BO algorithm and the associated set definitions are inspired by \safeopt~\cite{sui2015safe, berkenkamp2023bayesian}.
% Let us get the uncertainty tubes without contained set

The forthcoming safety guarantees rely only on sampling within a subset of the domain~$\domain$, namely the safe set~$S_t$, that contains parameters that are safe with high probability.
Specifically, we define the safe set~$S_t$ as
\begin{equation}\label{eq:safe}
    S_t \coloneqq %S_{t-1} %\cup \cap_{i\in\Ig} 
   % \cup_{a^\prime\in S_{t-1}} 
    %
    %\left(
    \cap_{i\in\Ig} 
    \{a \in \domain \mid \ell_{i,t}(a) \geq \underline h_i
    \},
    %\right),
\end{equation}
\ie $S_t$ contains all parameters~$a\in\domain$ whose lower bound~$\ell_{i,t}$ is greater than the safety threshold~$\underline h_i$ for all~$i\in\Ig$.
To initialize the optimization procedure, we assume that we have access to a non-empty initial safe set~$S_0\subseteq\domain$ with~$h_i(a)\geq\underline h_i$ for all~$a\in S_0, i\in\Ig$.
In addition to ensuring safety, we aim to efficiently balance exploration and exploitation while solving~\eqref{eq:opt}.
Hence, we introduce the set of potential maximizers $M_t \coloneqq \{
    a\in S_t \mid u_{0,t}(a) \geq \max_{a^\prime\in S_t} \ell_{0,t}(a^\prime)
    \},$
and the set of potential expanders
$G_t\coloneqq \{
a\in S_t\setminus M_t \mid g_t(a) > 0
\}$, with
$g_t(a)\coloneqq \lvert \{
a^\prime\in\domain\setminus S_t \! \mid \! a\in\argmin_{\widetilde a\in S_t} \|\widetilde a-a^\prime\|
\}\rvert$.
The set of potential maximizers~$M_t\subseteq S_t$ contains parameters that are likely to maximize the reward function, while evaluating within~$G_t$ may lead to an expansion of~$S_t$.
To reduce uncertainty in the candidate subset~$M_t\cup G_t$, we choose to evaluate the most uncertain parameter therein, \ie
\begin{align}\label{eq:acquisition}
    a_{t+1}=\argmax_{a\in M_t \cup G_t} \max_{i\in\I} \left(u_{i,t}(a)-\ell_{i,t}(a)\right).
\end{align}

\begin{algorithm}
\begin{algorithmic}[1]
    \Require $\kappa$, $\nu$, $\P_\mathrm{P}, \P_\epsilon$, $a_1\in S_0$, $y_1$, $T$
    \For{$t=1,\ldots,T$}
    \State $u_t,\ell_t \gets$ Algorithm~\ref{alg:wait}%($t,\kappa, \nu, \P_\mathrm{P},\P_\epsilon, a_{1:t}, y_{1:t}$)
% \If{$t>1$} $S_t\gets$~\eqref{eq:safe} \textbf{else} $S_t\gets S_0$
% \EndIf
\State Compute safe set~$S_t$ \Comment{\eqref{eq:safe}}
    \State Compute candidate sets~$M_t, G_t$ %\Comment{\eqref{eq:maximizer}, \eqref{eq:expander}}
\If{$\lvert M_t\cup G_t\rvert > 0$} $a_{t+1} \gets$ \eqref{eq:acquisition}  \textbf{else} break 
\EndIf
\State $y_{i,t+1}\gets h_i(a_{t+1}) + \epsilon_{i,t+1}$ \Comment{Conduct experiment}
    \EndFor
    \State \Return $\argmax_{a\in S_t}\ell_{0,t}(a)$ \Comment{Best safe parameter}
\end{algorithmic}
    \caption{Safe BO via function-based UQ}
    \label{alg:safe_BO}
\end{algorithm}

Algorithm~\ref{alg:safe_BO} summarizes our safe BO algorithm. 
First, we compute the uncertainty tubes~$u_t$ and~$\ell_t$ via the wait-and-judge scenario approach in Algorithm~\ref{alg:wait} (\lline~2).
Then, we compute the safe set~$S_t$ (\lline~3), the set of potential maximizers~$M_t$, and the set of potential expanders~$G_t$ (\lline~4).
We acquire the next parameter~$a_{t+1}$ from~\eqref{eq:acquisition} (\lline~5) and conduct an experiment (\lline~6).
The algorithm terminates if the maximum number of iterations~$T\in\mathbb N$ is reached (\lline~1) or if the candidate set~$M_t\cup G_t$ is empty (\lline~5).
Finally, we return the parameter that we believe safely maximizes the reward function (\lline~7).

Next, we prove that we satisfy the constraints of~\eqref{eq:opt} with high probability while executing Algorithm~\ref{alg:safe_BO}. % with high probability.
 
\vspace*{0.1cm}
\begin{theorem}[Per-iteration safety]\label{th:safety}
    Let Assumptions~\ref{asm:noise} and~\ref{asm:prob} hold and a set of initial safe parameters~$\emptyset\neq S_0\subseteq \domain$ be given.
    Further, choose any probability parameter~$\nu\in (0,1)$, any confidence parameter~$\kappa\in (0,1)$.
    Then, at any~$t\geq 1$, with confidence at least~$1-\kappa$, while executing Algorithm~\ref{alg:safe_BO}, 
        $h_i(a_{t+1})  \geq \underline h_i$
    holds with probability at least~$1-\nu$ for all~$i\in\Ig$. 
\end{theorem}
\begin{proof}
    % The proof is analogous to the proof of~\cite[Lemma~4.2]{berkenkamp2023bayesian}.
    Algorithm~\ref{alg:safe_BO} ensures that we only sample within the safe set~$S_t$. %~\eqref{eq:safe}. 
%For all~$a\in S_t$, we \AT{either have that~$\ell_{i,t}(a)\geq\underline h_i$ for all~$i\in\I$ or that~$a\in S_{t-1}$}. 
By definition of~$S_t$, we have~$\ell_{i,t}(a)\geq\underline h_i$ for all~$a\in S_t, i\in\Ig$, \ie the lower bounds of the uncertainty tubes~$\ell_{i,t}$ obtained from Algorithm~\ref{alg:wait} are above the safety threshold.
%The uncertainty tubes from Algorithm~\ref{alg:wait} are equipped with the guarantees from Theorem~\ref{th:uncertainty_tubes}.
%From Theorem~\ref{th:uncertainty_tubes}, we get at any iteration~$t\geq 1$, with confidence at least~$1-\kappa$ that $\ell_{i,t}(a)\geq\underline h_i$ implies $h_{i}(a)\geq\underline h_i$ for all~$a\in\domain, i\in\Ig$ with probability~$1-\nu$.
By Theorem~\ref{th:uncertainty_tubes},
with confidence at least~$1-\kappa$, the inequalities
$h_i(a)\geq\ell_{i,t}(a)$ hold simultaneously for all
$a\in\domain$ and $i\in\Ig$ with probability at least~$1-\nu$.
Since $a_{t+1}\in S_t$, we have
$\ell_{i,t}(a_{t+1})\geq\underline h_i$ for all $i\in\Ig$,
which proves the claim.
\end{proof}

\begin{comment}
\begin{remark}[Nature of safety guarantees]
    Note that Theorem~\ref{th:safety} establishes \emph{per-iteration} safety, \ie the probabilistic guarantees only hold for each \emph{fixed} iteration~$t\geq 1$.
To recover the standard \safeopt\ guarantees~\cite{sui2015safe, berkenkamp2023bayesian}, the uncertainty tubes in Theorem~\ref{th:uncertainty_tubes} would need to hold \emph{simultaneously} for all~$t\geq 1$.
Establishing such a result requires further analysis, as the measure~$\P$ depends on the observations~$\mathcal D_t$ and thus varies across iterations, \AT{which makes the use of a classical union bound not straightforwardly applicable.}
% Hence, a classical union bound as in \eg \cite[Theorem~1]{tokmak2026safe} is not straightforwardly applicable.
% We here therefore employ a heuristic confidence tightening~$\kappa_t$ and demonstrate in the following section that safety is indeed maintained while executing Algorithm~\ref{alg:safe_BO}.
\end{remark}
\end{comment}

\section{Experiments}\label{sec:exp}
We evaluate Algorithm~\ref{alg:safe_BO} on a synthetic example in Section~\ref{sec:toy} and demonstrate its successful end-to-end deployment by tuning the control parameters of a real Furuta pendulum in Section~\ref{sec:furuta}. 
We illustrate a discontinuous function class along with the corresponding uncertainty tubes obtained via Algorithm~\ref{alg:wait} in Section~\ref{sec:admissible}. 
All experiments were conducted with~$\kappa=10^{-3}$,~$\nu=10^{-1}$,~$\domain=[0,1]$, $N=1000$,~$\lvert\I\rvert=1$, and on an Ubuntu laptop with 32 GB RAM and an Intel Core i7-12700H, unless stated otherwise.

\subsection{Safe BO with synthetic examples}\label{sec:toy}
We consider the same function~$h_0$ as outlined in Example~\ref{ex:1} and set the safety threshold to~$\underline h_0=-2.15$, which corresponds to the~$20\%$-quantile of~$h_0$.
We then solve Problem~\eqref{eq:opt} with Algorithm~\ref{alg:safe_BO}.
The result is illustrated in Figure~\ref{fig:safe_BO_synthetic}.
We initialize Algorithm~\ref{alg:safe_BO} with~$S_0=\{0.53\}$ (magenta diamond), sequentially conduct~$T=30$ experiments (black samples), and return the best safe parameter (cyan square) without violating the safety constraints (dashed red line). %, which is in accordance with Theorem~\ref{th:safety}.
This particularly showcases the effectiveness of the uncertainty tubes (gray lines) as outlined in Algorithm~\ref{alg:wait} and Theorem~\ref{th:uncertainty_tubes}, which contain~$h_0$ uniformly over the domain at both~$t=1$ (upper sub-figure) and~$t=30$ (lower sub-figure).

\begin{figure}
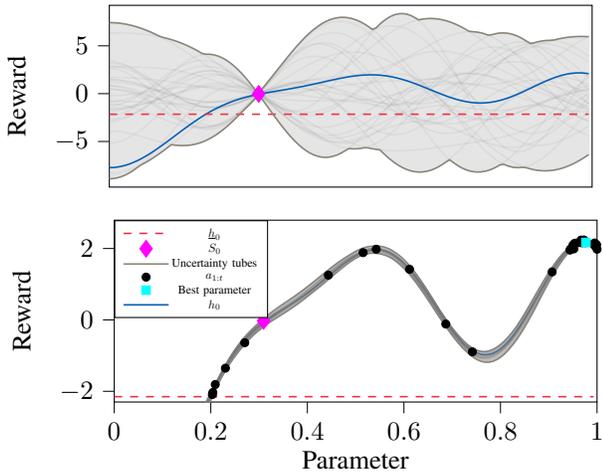

\begin{center}
    \input{Figures/final/numerical_example_beginning_NEW} \newline
\phantom{123}
    \input{Figures/final/numerical_example_end_NEW}
\end{center}
        \label{fig:synthetic_end}
    \caption{\emph{Synthetic safe BO example.} 
    We start (upper sub-figure) by executing Algorithm~\ref{alg:safe_BO} with one initial sample~$S_0$. 
    After $T=30$ iterations (lower sub-figure), Algorithm~\ref{alg:safe_BO} explored the domain, identified the global maximizer, and did not conduct an unsafe experiment. 
    }
    \label{fig:safe_BO_synthetic}
\end{figure}

% \begin{figure}
%     \centering
%     \caption{Caption}
% \end{figure}

\subsection{Hardware experiment}\label{sec:furuta}

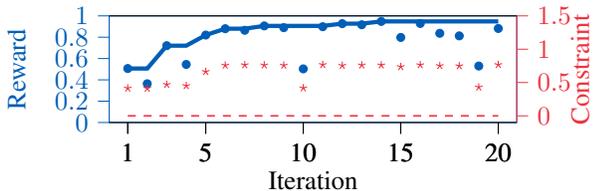
\begin{figure}
    \centering
    % This file was created with tikzplotlib v0.10.1.
\begin{tikzpicture}

\definecolor{darkgray176}{RGB}{176,176,176}

% \begin{axis}[
% tick align=outside,
% tick pos=left,
% x grid style={darkgray176},
% xlabel={Iteration},
% xmin=-0.95, xmax=19.95,
% xtick style={color=black},
% y grid style={darkgray176},
% ylabel=\textcolor{aaltoBlue}{Reward},
% ymin=0, ymax=1,
% height=3cm,
% width=6cm,
% ytick style={color=black}
% ]
\begin{axis}[
tick align=outside,
tick pos=left,
x grid style={darkgray176},
xlabel={Iteration},
xmin=-0.95, xmax=19.95,
xtick style={color=black},
y grid style={darkgray176},
ylabel=\textcolor{aaltoBlue}{{Reward}},
ymin=0, ymax=1,
ytick style={color=black},
height=3cm,
width=7cm,
  ytick style={color=aaltoBlue},
  yticklabel style={color=aaltoBlue},
  y label style={color=aaltoBlue},
  y axis line style={aaltoBlue},
  xtick={0,4,9,14,19},
xticklabels={1, 5,10,15,20},
]
\addplot [draw=aaltoBlue, fill=aaltoBlue, mark=*, only marks, mark size=1.5]
table {%
0 0.506385684013367
1 0.363650023937225
2 0.720545709133148
3 0.544653236865997
4 0.82013988494873
5 0.879775762557983
6 0.864942491054535
7 0.906046986579895
8 0.88959676027298
9 0.502953767776489
10 0.898508489131927
11 0.927131056785583
12 0.91692715883255
13 0.948114693164825
14 0.797857046127319
15 0.928743720054626
16 0.836684107780457
17 0.812491238117218
18 0.529780328273773
19 0.881866991519928
};
\addplot [ultra thick, aaltoBlue]
table {%
0 0.506385693565223
1 0.506385693565223
2 0.720545729466762
3 0.720545729466762
4 0.820139858079972
5 0.879775733289558
6 0.879775733289558
7 0.906046983327039
8 0.906046983327039
9 0.906046983327039
10 0.906046983327039
11 0.927131057241255
12 0.927131057241255
13 0.948114671928994
14 0.948114671928994
15 0.948114671928994
16 0.948114671928994
17 0.948114671928994
18 0.948114671928994
19 0.948114671928994
};
\end{axis}

\begin{axis}[
axis y line=right,
tick align=outside,
x grid style={darkgray176},
xmin=-0.95, xmax=19.95,
xtick pos=left,
xtick style={color=black},
y grid style={darkgray176},
ylabel=\textcolor{aaltoRed}{{Constraint}},
ymin=-0.1, ymax=1.5,
ytick pos=right,
ytick style={color=black},
yticklabel style={anchor=west},
height=3cm,
width=7cm,
  ytick style={color=aaltoRed},
  yticklabel style={color=aaltoRed},
  y label style={color=aaltoRed, yshift=2ex},
  y axis line style = {-, aaltoRed},
    xtick={0,4,9,14,19},
xticklabels={1, 5,10,15,20}
  ]

\addplot [draw=aaltoRed, fill=aaltoRed, mark=star, only marks, mark size=1.5]
table {%
0 0.414174824953079
1 0.4111068546772
2 0.469398111104965
3 0.450990349054337
4 0.662679672241211
5 0.757786512374878
6 0.763922452926636
7 0.760854482650757
8 0.757786512374878
9 0.417242765426636
10 0.766990423202515
11 0.754718542098999
12 0.760854482650757
13 0.760854482650757
14 0.73631078004837
15 0.763922452926636
16 0.75165057182312
17 0.748582601547241
18 0.429514616727829
19 0.766990423202515
};
\addplot [semithick, aaltoRed, dashed]
table {%
0 0
1 0
2 0
3 0
4 0
5 0
6 0
7 0
8 0
9 0
10 0
11 0
12 0
13 0
14 0
15 0
16 0
17 0
18 0
19 0
};
\end{axis}

\end{tikzpicture}
    \caption{\emph{Scaled reward and constraint development  (hardware experiment).}
The blue points denote the observed reward.
The blue line shows the running maximum.
The red stars show the minimum of both constraint values, which always remains above the safety threshold (dashed red line).
}
    \label{fig:reward_development_twin}
\end{figure}

\begin{figure}
    \centering
    \input{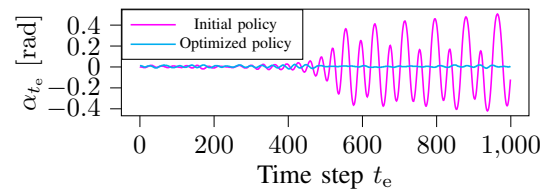}
    \caption{\emph{Trajectory of the pendulum angle~$\alpha_\te$.} 
    % The upright position of the pendulum corresponds to~$\alpha_\te=0$.
    The optimized policy oscillates significantly less around the upright position of the pendulum ($\alpha_\te=0$) than the initial policy.
    }
    \label{fig:angles}
\end{figure}

Next, we validate our safe BO algorithm (Algorithm~\ref{alg:safe_BO}) with its corresponding uncertainty tubes (Algorithm~\ref{alg:wait}) in a hardware experiment. % experiment on a real Furuta pendulum~\cite{furuta1992swing}.
Specifically, we safely tune the control parameters of a Furuta pendulum.
We consider a linear state-feedback controller of the form~$u_\te=K_t x_\te$, with~$x_\te = [\theta_\te,\alpha_\te,\dot{\theta}_\te,\dot{\alpha}_\te]^\top$, where~$\te\in [\Te]$ denotes the discrete time index within one episode,\footnote{%
We distinguish between two time indices; the index~$\te$ denotes the discrete time step within a single episode (\ie within one roll-out of the Furuta pendulum), whereas~$t$ denotes the iteration index for safe BO, with each~$t$ corresponding to one completed episode.
} $\Te\in\mathbb N$ denotes the episode horizon, and $\theta_\te,\alpha_\te$ are the rotary arm and pendulum angles, respectively.
We follow the setup described in~\cite{baumann2021gosafe}.
In particular, we tune the first two entries of the feedback vector~$K_t$ and solve~\eqref{eq:opt} with the reward function~$h_0$ promoting stabilization around the upright position and the constraint functions~$h_1,h_2$ enforcing angle limits along the trajectory, 
\begin{align*}\label{eq:hardware_functions}
h_0(a_t)&=\frac{1}{T_\mathrm{e}}\sum_{t_\mathrm{e}=1}^{T_\mathrm{e}}
\max
\left\{
1-\frac{0.8\lvert \alpha_{t_\mathrm{e}}\rvert + 0.2\lvert\theta_{t_\mathrm{e}}\rvert}{\pi}
, 0
\right\}^2 \\
h_1(a_t)&=
\frac{\pi}{2} - \max_{\te\in [\Te]}\{
\lvert \theta_\te 
\rvert
\},
\quad
h_2(a_t)=  \frac{\pi}{4} - \max_{\te\in [\Te]}\{
\lvert \alpha_\te 
\rvert
\},
\end{align*}
with~$\underline h_1=\underline h_2=0$.
We discretize the domain~$\domain=[0,1]^2$ into~$N=\num{10201}$ equidistant points, set the episode horizon to~$\Te=\num{1000}$, and assume that the inherent measurement noise (Assumption~\ref{asm:noise}) is uniformly distributed, \ie $\epsilon_{i,t}\sim \mathcal U(-\delta,\delta), \delta=5\times 10^{-2}$.
For the reward and constraint functions~$h_i$, $i\in\I=\{0,1,2\}$, we assume that they admit a representation as in Assumption~\ref{asm:prob} with basis functions~$\varphi_{i,r}(a)=k_{\mathrm{Ma}32}(a,a_r)$, where~$k_{\mathrm{Ma}32}$ denotes the Matérn32 kernel with lengthscale $10^{-2}$, and we assume zero-mean Gaussian coefficients with standard deviation $0.1$.
Moreover, we initialize our safe BO algorithm (Algorithm~\ref{alg:safe_BO}) with an initial safe set~$S_0=\{[0.23,0.40]^\top\}$ and optimize for a total of~$T=20$ iterations. 
The exploration behavior is illustrated in Figure~\ref{fig:Furuta_intro}.
Starting from an initial parameter (magenta), the safe set (green hull) is progressively expanded by evaluating different parameters (black points), ultimately identifying the best parameter (cyan).
As depicted in Figure~\ref{fig:reward_development_twin}, the reward increases without conducting unsafe experiments.
Figure~\ref{fig:angles} demonstrates the reward increase, where the pendulum angle~$\alpha_{t_e}$ exhibits significantly reduced oscillations around the upright position under the optimized policy (cyan) compared to the initial policy (magenta).

Note that we consider an episodic parameter-tuning setting, in which each iteration of safe BO corresponds to one closed-loop roll-out of the pendulum. 
Hence, neither the computation of the uncertainty tubes in Algorithm~\ref{alg:wait} nor the safe BO procedure in Algorithm~\ref{alg:safe_BO} is required to run within the sampling period of the closed-loop dynamics.

\subsection{Function-based UQ for more diverse functions}\label{sec:admissible}
Next, we demonstrate the flexibility of function-based UQ, as formalized in Section~\ref{sec:function-based-UQ}, and the reliability of our proposed uncertainty tubes (Algorithm~\ref{alg:wait}, Theorem~\ref{th:uncertainty_tubes}) for a discontinuous function. %, 
We thereby show that the potential applications extend far beyond well-behaved functions and safe control parameter tuning.
We consider a heavy-tailed coefficient distribution~$c_\mathrm{P}\sim 10^{-2}\cdot \mathcal T_{10}$, where~$\mathcal T_{10}$ denotes a Student's-$t$ distribution with ten degrees of freedom, and choose the basis functions~$\varphi_{0,r}(a)$ as the Haar wavelet basis on the domain~$\domain$~\cite{mallat1999wavelet}.
The observation noise is set to be zero-mean Gaussian with standard deviation $10^{-2}$.
Figure~\ref{fig:wavelet} illustrates the uncertainty tubes obtained with Algorithm~\ref{alg:wait} after one function evaluation at~$a_{1}=10^{-2}$ (left) and after fifteen evaluations~$a_{1:t}\in[0, 0.02]$ (right), covering~$h_0$ and tightening around the sampled points.

\begin{figure}
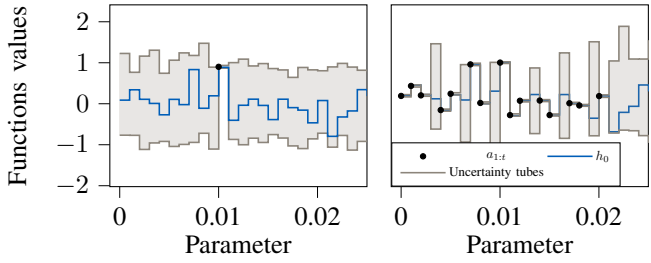

    \centering
    \input{Figures/final/Haar_wavelet_beginning}
    \input{Figures/final/Haar_wavelet_end}
    \caption{\emph{Haar wavelet basis.}
    Note that the~$x$-axis is zoomed in to illustrate the effect of the samples on the uncertainty tubes, and the support scenarios are omitted for clarity.
}
    \label{fig:wavelet}
\end{figure}

\section{Concluding remarks}\label{sec:conclusion}
In this paper, we introduced function-based UQ with applications to safe learning-based control.
Function-based UQ models the unknown function as a random function that admits a linear expansion in user-chosen basis functions with random coefficients, from which i.i.d.\ realizations can be sampled~\ref{co:asm}.
These realizations give rise to random functions, which we refer to as \emph{scenarios}, and enable the construction of uncertainty tubes around the unknown function with high probability via the scenario approach and the wait-and-judge framework~\ref{co:bounds}.
Our uncertainty tubes do not require a bound on, \eg the RKHS norm, and accommodate a wide range of functions, including discontinuities.
We integrated these uncertainty tubes into an existing safe BO algorithm with which we safely tuned control parameters of a real Furuta pendulum~\ref{co:exp}.
Future work includes investigating scalability, extending function-based UQ beyond safe BO, and analyzing its robustness to model misspecification. 
A comparison between the proposed scenario-based uncertainty tubes and alternative UQ methods, such as conformal prediction~\cite{lindemann2024formal}, is also of interest.

% \section*{Acknowledgments}
% AT thanks Simone Garatti for valuable discussions. % and insightful feedback.

\bibliographystyle{IEEEtran}
\bibliography{IEEEabrv,mybibfile}

\end{document}